\pdfoutput=1 
\ifdefined\IsLLNCS
\documentclass{llncs}
\usepackage{amsfonts,amsmath,amssymb,boxedminipage,color,url,nccmath}
\else
\documentclass[11pt]{article}
\usepackage{amsfonts,amsmath,amssymb,boxedminipage,color,url,fullpage,nccmath}
\fi
\usepackage{enumerate,paralist}
\usepackage[hypcap=false]{caption}
\usepackage{lmodern}
\usepackage[bottom]{footmisc} 
\usepackage[pageanchor=false,pdfstartview=FitH,colorlinks,linkcolor=blue,filecolor=blue,citecolor=blue,urlcolor=blue]{hyperref} 
\ifdefined\IsLLNCS\else
\usepackage{amsthm} 
\fi
\usepackage{aliascnt}
\usepackage[numbers]{natbib} 
\usepackage{cleveref}
\usepackage{xspace}
\usepackage{xstring}
\usepackage{float}
\usepackage{color,soul}
\usepackage{tikz}\usetikzlibrary{arrows}
\usetikzlibrary{decorations.markings}
\usepackage{subcaption}
\ifdefined\IsLLNCS
\fi
\usepackage{algorithm}
\usepackage[noend]{algpseudocode}
\usepackage{nicefrac}

\newtheorem{theorem}{Theorem}
\newtheorem*{theorem*}{Theorem}
\newtheorem{corollary}{Corollary}
\newtheorem{lemma}{Lemma}
\newtheorem{claim}{Claim}

\newtheorem{remark}{Remark}

\begin{document}
	
\title{Resilience of Networks to Spreading Computer Viruses: Optimal Anti-Virus Deployment (Extended Version\footnote{A shorter version of this paper was published at IEEE/IFIP Network Operations and Management Symposium (NOMS), 2023, titled \textit{"Resilience of Networks to Spreading Computer Viruses: Optimal Strategies for Anti-Virus Deployment"} \cite{Tavori2023NOMS}. This extended version includes additional analyses, detailed proofs, and expanded experimental results.} )}

\author{
Jhonatan Tavori
\thanks{Blavatnik School of Computer Science, Tel Aviv University, Israel. E-mail: \texttt{jhonatan.tavori@cs.tau.ac.il}}
\and 
Hanoch Levy
\thanks{Blavatnik School of Computer Science, Tel Aviv University, Israel. E-mail: \texttt{hanoch@tauex.tau.ac.il}}}

\date{}

\maketitle

\begin{abstract}

Deployment of anti-virus software is a common strategy for preventing and controlling the propagation of computer viruses and worms over a computer network. As the deployment of such programs is often limited due to monetary or operational costs, devising optimal strategies for their allocation and deployment can be of high value to the operation, performance, and resilience of the target networks.

We study the effects of anti-virus deployment (i.e., “vaccination”) strategies on the ability of a network to block the spread of a virus. Such ability is obtained when the network reaches “herd immunity”, achieved when a large fraction of the network entities is immune to the infection, which provides protection even for entities which are not immune. We use a model that explicitly accounts for the inherent heterogeneity of network nodes activity and derive optimal strategies for anti-virus deployment.

Numerical evaluations demonstrate that the system performance is very sensitive to the chosen strategy, and thus strategies which disregard the heterogeneous spread nature  may perform significantly worse relatively to those derived in this work.
\\
\\
\textbf{Keywords:} Network Resilience; Stochastic Spreading; Computer Viruses; Resources Allocation; Convex Optimization;

\end{abstract}

\section{Introduction}

Computer viruses, worms and other self-replication malicious programs form a persistent security threat on the Internet since the “Creeper” virus propagated through the ARPANET in 1971 \cite{afianian2019malware}. Such malwares take advantage of the communication between the network entities to spread themselves, in an automated fashion, thus allowing infected computers to infect susceptible computers \cite{hansman2005taxonomy,chen2007optimal,kim2004measurement}. Their damage potential is crucial as they can potentially spread across large-scale networks within short periods of time. Virus attacks disrupted hundreds of millions of systems, and cost both the public and private sectors billions of dollars \cite{neubauer2002protection}.

The analysis and characterization of viral spreading processes is important both to evaluate their potential damage, and for the devise of efficient counter measures to be taken by individuals or a network administrator \cite{yang2012computer,peng2013smartphone,neubauer2002protection}.

In this work we focus on studying strategies for the deployment of anti-virus programs (or scanners, HIDS, etc.) over the devices in a network. Installations of anti-virus software, which we refer to as “vaccination mechanisms”, is commonly known for protecting individual hosts against viruses. Nonetheless, as demonstrated in the recent COVID-19 crisis, a vaccination not only protects the individual computer being vaccinated, but it also contributes to protecting the entire network by contributing to reducing the network’s viral spread.

Our interest is in an operator who faces a multi-region environment, {such as an organization consisting of multiple international separate sites, or a university consisting of multiple (weakly connected) departments}. Such an operator aims at optimizing the deployment of a limited number of anti-virus licenses over these sites.  The limitation on the number of licenses stems from budget or human resources limitations or simply from the difficulty of getting all users to cooperate and install the program on their equipment. 

Given a limited supply of anti-virus program licenses, how should the licenses be distributed over a multi-region computer system, in order to maximize the resilience of the network to a potential attack? 
Should the efforts be concentrated,
namely, select a small set of regions and invest the entire budget and efforts on allocating and installing anti virus programs on their devices to reach full protection on the selected group;  or, rather, be scattered, namely, divide the efforts across the system, providing partial protection to each region?  
Further, what would be the optimal deployment strategy in a system where the virus has already started spreading? Will the interleaving of the natural infection process with the vaccination process affect the deployment strategy? 
These are some of the questions that motivate this research and are at its core.

\subsection{Utilizing Pandemic Studies and Models}

Viral spreading processes are common to computer viruses and worms spreading over computer networks, information broadcasts or spread of rumors over social networks, and to infectious diseases spreading over the human network.

A common property shared by computers communication and humans’ interaction is that both are known to be characterized by heterogeneity of connectivity and activity. Such heterogeneity has been recognized and studied for quite a while; 
see e.g., \cite{chen2007optimal,kim2004measurement,peng2013smartphone, adiga2016delay,draief2006thresholds, adamic2002zipf, lloyd2001viruses,tuchner2020bullshit} in the computer networking literature, and \cite{pastor2015epidemic, rock2014dynamics, smith2005} in the human biology literature. 

The ongoing COVID-19 pandemic has made it possible to observe this heterogeneity phenomenon conspicuously, as heterogeneity of connectivity, and possibly other properties, translated to heterogeneity of infectiousness and susceptibility. This has been expressed in the form of "super-spreaders", a small fraction of the population which is responsible to a large fraction of the infection cases \cite{endo2020estimating}. Due to the recent high interest in pandemics, the effect of heterogeneity across the population on the progression of epidemic processes, the acquirement of herd immunity, and counter-measure strategies was subject to a recent research \cite{britton2020mathematical,oz2021heterogeneity, Tavori02012024}.

A striking result of recent studies has shown that heterogeneity across the network entities may have a drastic effect on the viral process progression and may reduce dramatically the number of infected individuals prior to reaching herd immunity. This dramatic effect is rooted in the behavior of the stochastic process in which super-spreading causes the process to progress in a biased way. Roughly speaking, the existence of super-spreaders (i.e., nodes with high level of infectiousness/susceptibility) yields that they stochastically tend to get infected and develop immunity or "leave the game" in an early stage of the epidemic process, leading to a decrease in the viral spread (rate of new infections) within the network.

In particular, it was shown \cite{oz2021heterogeneity,tavorilevyrconvex} that the well-known reproduction number $R$, which measures the viral spread of the attack within the network (defined as the expected number of secondary infections produced by an infected node) decreases in \textit{convex} manner as the virus spreads in heterogeneous populations. This is in contrast to the traditional \textit{linear} reduction of $R$ in homogeneous populations. Thus, in heterogeneous populations the viral spread drops drastically at early stages of the spread and, in comparison to homogeneous networks, the infection process comes to an earlier end, expressed by reaching herd immunity earlier (as $R$ drops below 1). 

{In light of these findings, we adopt the "herd immunity" term from the biological domain and apply it in the computer networks domain. The importance of herd immunity is that it marks the turnaround point of the virus spread in the network, as reaching it causes the number of new infections to decline and the spread gets under control. This turnaround point, which is termed Herd Immunity Threshold (HIT), can be derived analytically in many cases and we aim at minimizing it. In addition, we observe in our simulation experiments that deployment strategies which minimize the HIT minimize the overall number of infections as well.}

We analyze the stochastic progression of the process when the system is subject to anti-virus deployment (i.e., vaccination process). We aim at understanding how the biased progression of the spread, as revealed by \cite{britton2020mathematical,oz2021heterogeneity,tavorilevyrconvex} can be utilized to guide network administrators in limiting the scope of the malware spread within their network.

\subsection{Our Contributions and Structure of the Paper}
Our analysis accounts for an arbitrary spreading distribution, which allows a general level of heterogeneity across the network's entities. 

We assume that the \textit{statistics} of spreading level in the regions are available to the operator who decides on the quantities (and timing) of deployment in each region, while 
the spreading levels of \textit{individual} nodes are not available.
Of course, if one {knows} who the heavy spreaders are and can target them, then the trivial strategy of vaccinating nodes in the order of their spreading level seems to be optimal. Nonetheless, this targeting is not always possible as the interaction levels across the network may vary over time and such knowledge may not be in the hands of the operator.  Thus, we analyze deployments in which the allocation \textit{within} a region is done in a random fashion across its nodes. 

We begin (Section \ref{sec:sec3}, following a formal presentation of the model in Section \ref{sec:sec2}) with analyzing a setting where the anti-virus deployment occurs prior to the beginning of the virus attack ("offline" scenario). The optimal strategy aims at  minimizing the future damage of the attack when it is launched, or even at blocking it.
We prove that the marginal benefit from each additional anti-virus installed in the region is decreasing in the total number of installations deployed. This implies that the optimization problem is a convex one, and thus the optimal deployment policies are not concentrated on specific regions.
We use the convexity to propose an efficient algorithm which, based on the spreading distributions of the regions,  allocates a given supply of licenses among the regions. The algorithm works in a greedy manner and yields the optimal solution efficiently.

In Section \ref{sec:sec4} we address the optimal deployment strategy for a system in which the virus attack has already started spreading ("online" scenario). 
In many practical cases the deployment of anti-virus licenses (or part of it) must be carried out while the virus attack is on. This yields a stochastic process that consists of the interleaving of the viral spreading process and of the anti-virus deployment process. Thus, to derive optimal deployment one must analyze this combined process.

An intriguing question, with no obvious answer, is whether delaying vaccinations may be beneficial. On the one hand the basic intuition says that pushing the vaccination process early can only do good. On the other hand, one notices that the natural spreading process is biased towards “catching” heavy spreaders and “take them out of the game” early in the process, while the vaccination process is conducted uniformly (no bias towards heavy spreaders). Thus one might see an advantage of postponing the vaccination process to let the natural spreading “catch” the heavy spreaders.  

We analyze the combined process and establish a monotonicity property: any delay in the deployment of a given batch of anti-virus installations, will result with an increased number of infections prior to reaching herd immunity.
Thus, one should not postpone vaccinations. Using the monotonicity property we apply the greedy approach of Section \ref{sec:sec3} and propose a deployment strategy for the online scenario.

In Section \ref{sec:sec5} we use simulative experiments to evaluate the sensitivity of the system to the chosen deployment strategy. We demonstrate that the number of infections prior to reaching herd immunity,
and the overall number of infected nodes, 
is very sensitive to the chosen allocation, especially due to heterogeneity. We demonstrate that allocations which disregard heterogeneity may perform significantly worse (up to tens of percents) relatively to heterogeneity-accounting allocations (like the one proposed in this paper). 
Further work and concluding remarks are given in Section \ref{sec:sec6}.

\section{{The Model and Prior Fundamental Results}}\label{sec:sec2}

In this section we provide a formal description of the spreading process, along with a short review of prior fundamental results that our analysis relies upon.

The \textit{Susceptible-Infective-Recovered} (SIR) model is a standard model of epidemic processes \cite{lloyd2001viruses,pastor2015epidemic,newman2003structure}, and has been used to model the spread of computer viruses and worms as well \cite{peng2013smartphone,amador2014stochastic}.

We adopt this model with the extension of \textit{Vaccinations}, SVIR, as follows:
Initially, all nodes are assumed to be susceptible (\textbf{S}). 
As a result of an \textit{infection}, a susceptible node become infective (\textbf{I}). 
During their infectious period, susceptible nodes may infect other nodes they contact with.
After the infectious period, infected nodes "leave the game", i.e., remove/recover (\textbf{R}) and {stop spreading the virus} (e.g., being patched or as the infection mechanism gets "exhausted" after trying to infect all its neighbors or a timeout has elapsed).
Additionally, susceptible nodes can move directly to being vaccinated (e.g., by deploying an anti-virus program) (\textbf{V}), meaning that they can not get infected or infect others. This is demonstrated in Figure 1.

We will use the terms "vaccine" and "anti-virus" (or "anti -virus license") interchangeably.
\begin{figure}[h]%
	\centering
	\subfloat
	{{\includegraphics[width=0.45\linewidth]{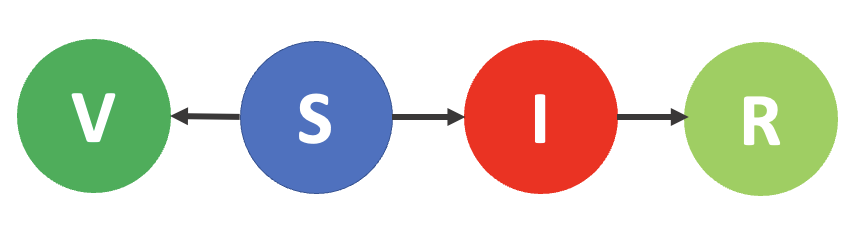} }}%
	\caption{\textbf{S}usceptible-\textbf{I}nfective-\textbf{R}ecovered or \textbf{V}accinated Model.}%
	\label{fig:example}%
\end{figure}
\subsection{Heterogeneity of Spreading}

We consider a network of nodes $V$, of size $\vert V \vert = N$.
At the beginning of the process, each node $u \in V$ is assigned with susceptibility and infectiousness parameters (or weights), $S(u)$ and $I(u)$;
\footnote{Practically -- these values can be produced by sampling a population. If one wants to attribute them to a specific theoretical distribution, say Gamma (,) one may draw each of the values from that distribution and obtain $\mathcal{S}$ and $\mathcal{I}$ that approximate the theoretical distribution very well (especially in practical situations where population sizes are in millions).}.
drawn from $\mathcal{S}$ and $\mathcal{I}$. 
these values accompany $u$ throughout the entire process.
The ensemble of the values over all $u \in V$ forms the \textit{susceptibility} and \textit{infectiousness} \textit{distributions}, or \textit{statistics}, of the network, $\mathcal{S}$ and $\mathcal{I}$ (where $ \text{Supp} (\mathcal{S}) = \text{Supp}(\mathcal{I}) = [0,1]$).

Assuming that $u_i \in V$ is infective, the probability that it will infect the susceptible node $u_j \in V$ is:
\begin{equation}
	\Pr[\text{$u_i$ infects $u_j$}] = I(u_i) \cdot S(u_j).
\end{equation}
$I(u_i)$ represents the probability of $u_i$ to spread the virus to others when it is in infective mode; $S(u_j)$ represents the probability of $u_j$ to attract the infection. 

The interaction network can be viewed as a random graph, where 
edge from $u_i$ to $u_j$ is present with probability $I(u_i)\cdot S(u_j)$, and $I(u)$ and $S(u)$ are proportional to the expected out-degree and in-degree of $u$.

Note that $S(u)$ and $I(u)$ implicitly include the interaction level of $u$ (probability of communicating with other nodes) as well as any other properties of $u$ (such as its software properties, etc.). Thus, throughout this work we allow those distribution to be arbitrary, in order to incorporate any level of heterogeneity across the network entities.
The only assumption we take on $\mathcal{S}$ and $\mathcal{I}$ is that more susceptible nodes are stochastically more infectiousness\footnote{Formally, we assume that their sampling is positively correlated such that $\varphi(s) := \mathop{\mathbb{E}} _{u \in V}  \left[ I(u) \: \vert \: S(u) = s \right]$  is monotonically non-decreasing in $s$. Of course, the case of equality $I(u)=S(u)$ for any $u \in V$ is a special case of this definition. The case where the infectiousness is independent of susceptibility (i.e., $\varphi(s_1) = \varphi(s_2)$ for any $s_1 \neq s_2$) is a special case as well.} (a lighter assumption than assuming that $S$ and $I$, the in- and out- degrees, are equal).

\subsection{The Effective Reproduction Number and Herd Immunity}

In order to evaluate the propagation rate of the spread, we track the reproduction number. The \textit{basic} reproduction number, $R_0$, is a a measure of how transferable the virus is in the network. 
It is defined as the expected number of secondary cases produced directly by the first infected node. 

As the spread continues, varying proportions of the population are recovered/removed (\textbf{R}) or vaccinated (\textbf{V}) at any given time, and the size of the susceptible population decreases. Thus, the expected number of secondary cases produced directly by an infected node will reduce over time.
The \textit{effective} reproduction number, $R(n)$, denotes the expected number of infections directly generated by the $n$th infected node\footnote{We index $R$ as a function of the number of nodes which contracted the virus. The event whereby the $n$th infection occurs is called the $n$th \textit{step} of the spread, and the number of steps is upper-bounded by $N$. This indexing method will be useful in our analysis since we aim at minimizing the number of infections prior to the herd immunity. }:
\begin{equation*}\label{eq:rorigin}
R(n) = \mathbb{E} \left[ \text{\# of infections generated by the $n$th infected node} \right]
\end{equation*}
where the expectation is taken over {all nodes and all possible scenarios of infections}. 

The decrease of the reproduction number is a key concept in SIR spreading processes.
As long as $R$ obeys $R>1$, then the number of infections increases in an exponential manner (which makes the spread much more difficult for the network operator to cope with).
When $R < 1$, the number of infection cases decreases, and the spreading process comes to an end (i.e., herd immunity is achieved). 
The fraction of the population that contracted the spread before $R$ reaches $1$ is called the \textit{Herd Immunity Threshold} (or HIT). 
Figure \ref{fig:example2} demonstrates herd immunity achieved by vaccinations. 

\begin{figure}[h]%
	\centering
	\subfloat
	{{\includegraphics[width=0.55\linewidth]{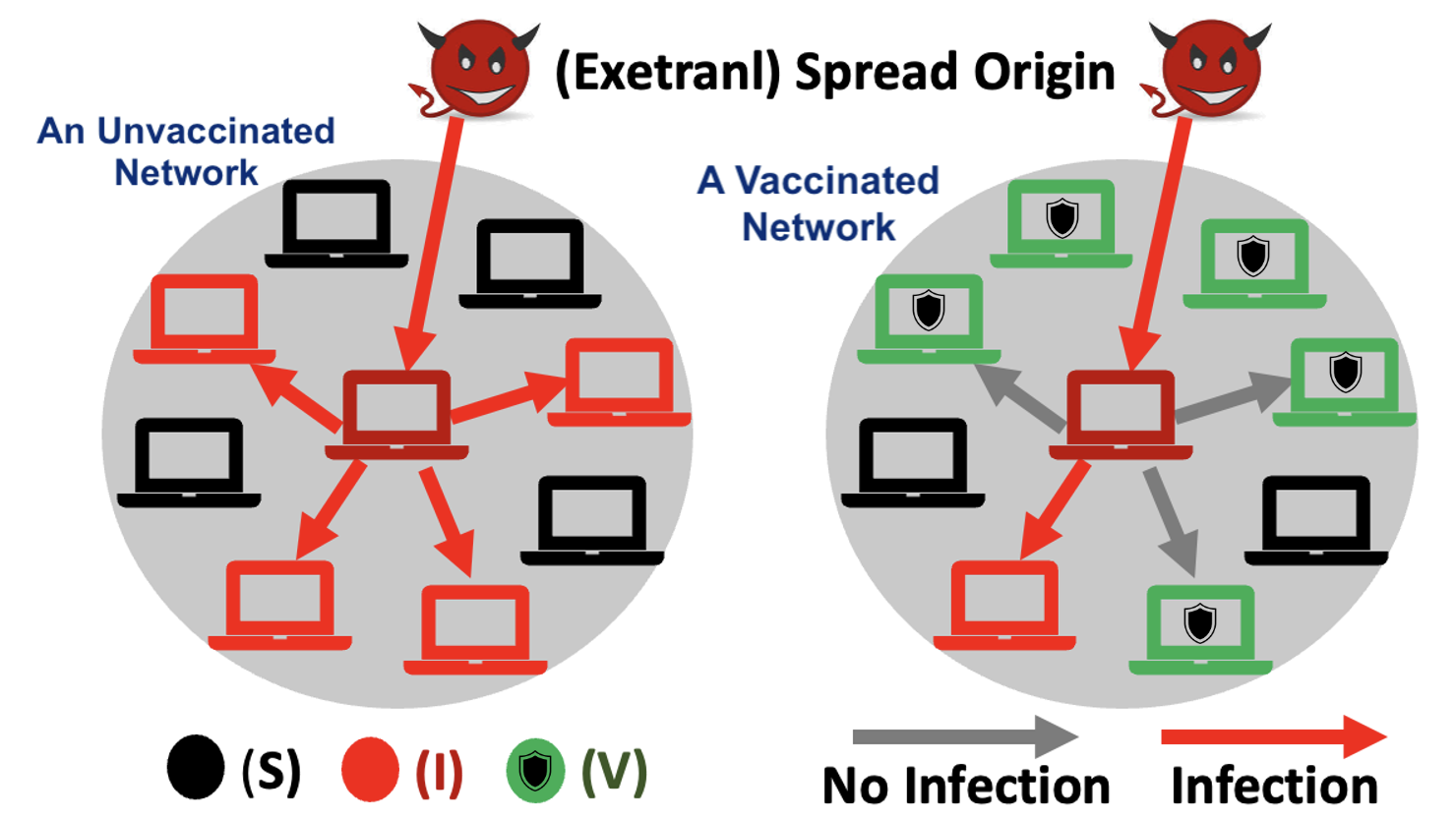} }}
	\caption{Herd immunity achieved by installing anti-virus programs (i.e., vaccination).}
	\label{fig:example2}
\end{figure}

\subsection{Prior {Fundamental} Results}
Our analysis will be based on a prior result \cite{tavorilevyrconvex} establishing that $R$ is convex. 
For the sake of completeness (but not needed to follow our analysis), in this sub-section we briefly describe the {assumptions} used in the model and explain where the convexity of $R$ stems from. 

In a homogeneous population, it was "traditionally" known to have a linear reduction of $R$ (i.e., when 50\% of the entities were immune, then $R$ was cut in half).  However, it was recently established that in heterogeneous populations $R$ decreases in a convex manner, see Figure 3. The reason is that the correlation between $I(u)$ and $S(u)$ implies that highly infective nodes are more susceptible and therefore get infected early in the process, "leave the game" and lead to sharp decrease in $R$ {in early stages}, and to the convexity of $R$.  

\begin{figure}[h]%
	\centering
	\subfloat
	{{\includegraphics[width=0.9\linewidth]{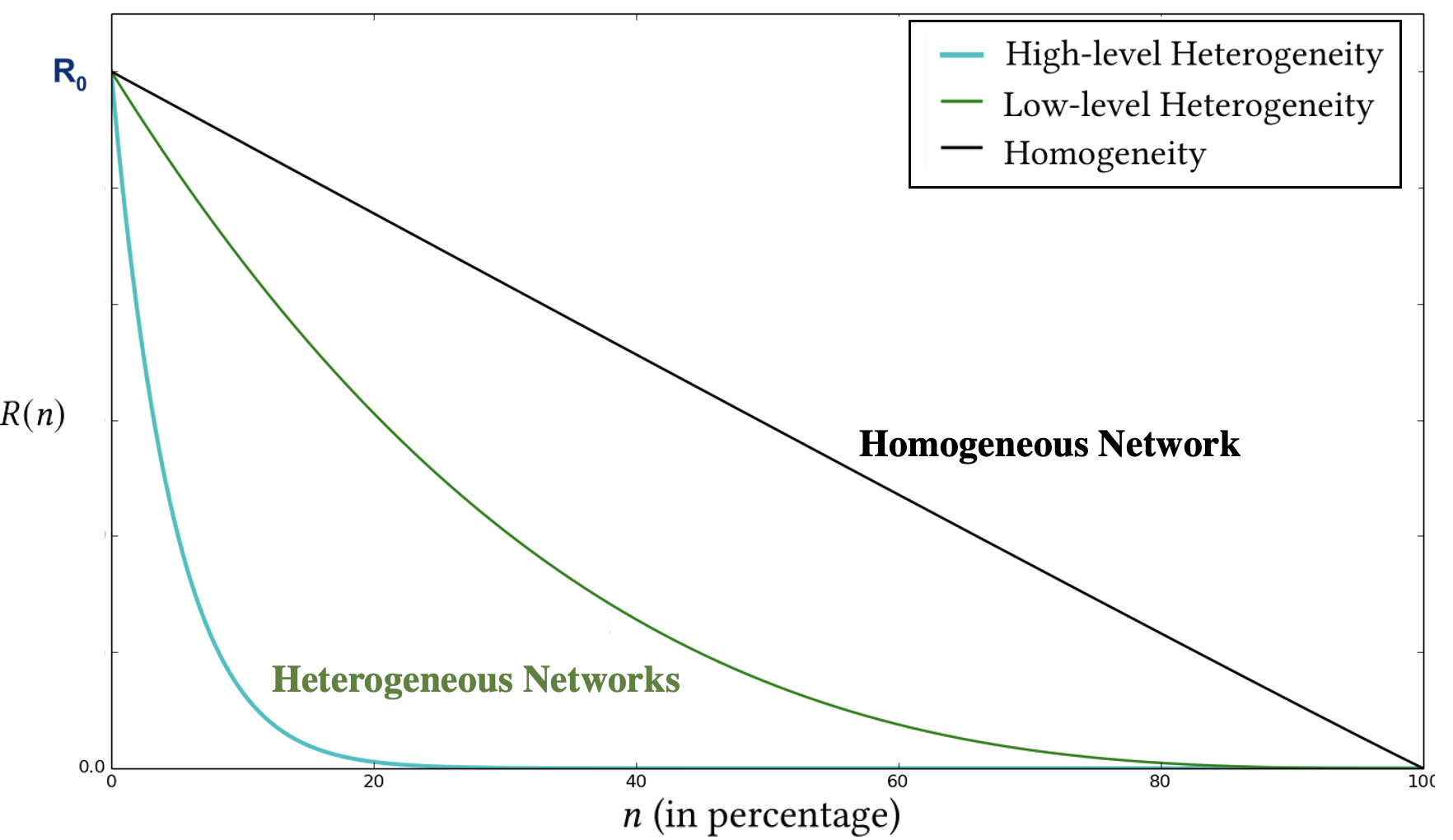} }}
	\caption{The reduction of $R$ in Homo/Hetero-geneous networks, as a function of $n$, the number of infections (in percentage).}
	\label{fig:example3}
\end{figure}

The {distribution} of the susceptibility and infectiousness parameters across the susceptible nodes may change throughout the spread, as infected/vaccinated nodes are removed from the susceptible population.
Let $\textbf{V}_n$ denote the susceptible nodes after $n$ infections (a subset of the susceptible nodes at step $n-1$). $\textbf{V}_n$ is a random variable distributed over all possible scenarios of infection. 

By tracking the probability density functions of $\textbf{V}_n$ for $n=0,1,\dots$, and establishing stochastic dominance relations between them (which reflect the fact that the heavy spreaders stochastically tend to get infected earlier than light spreaders), it was established that:

\begin{theorem}[Convexity of $R$ \cite{tavorilevyrconvex}]
	$R(n)$ is convex in $n$.
\end{theorem}

\section{{The effects of vaccinations prior to the spreading process}}\label{sec:sec3}
In this section we study the "offline" scenario where the deployment occurs prior to the beginning of the virus attack. In Section \ref{sec:sec4} we will analyze a deployment which occurs during the spreading process. 

{As discussed in the introduction}, 
we assume that the operator is only aware of the spreading level {statistics} in the different regions, and needs to decide the quantities to deploy in each region, while the deployment \textit{within} each region is done in an "oblivious" random fashion across its nodes.

We begin (Subsection \ref{subsec:32}) with proving that  
{within each region, the herd immunity threshold resulted from the vaccination process is convex in the number of vaccines deployed in the region (Theorem \ref{thm:hitconvex}).
In Subsection \ref{secvaccineconvex} we use this result and propose a greedy algorithm for deriving the optimal allocation over multiple regions.}

\subsection{Convexity of The HIT in The Number of Vaccines Deployed}\label{subsec:32}

{We consider a network of $N$ nodes in which $x$ anti-virus licences where deployed uniformly (i.e., $x$ random nodes were vaccinated). Afterwards, a virus has attacked the network. Denote by $H(x)$ the expected\footnote{The expectation is taken over all scenarios of infections.} number of nodes which were  infected prior to reaching herd immunity. We prove that:}
\begin{theorem} \label{thm:hitconvex}
	{The expected number of infections prior to reaching herd immunity, $H(x)$, is convex in the number of anti-virus deployed in the region, $x$, for any $x>0$.}
\end{theorem}

\begin{proof}[Proof of Theorem \ref{thm:hitconvex}]
As a result of the vaccination process, the number of susceptible nodes at the beginning of the virus attack is $N - x$.
Clearly, $H(x)$ is decreasing in $x$ (since the number of susceptible nodes is decreasing as $x$ increases).
In order to prove the theorem, we need to show that:
\begin{equation}\label{maineq}
	\forall x>0, \; H(x) - H(x+1) \ge H(x+1) - H(x+2).
\end{equation}
For convenience, we denote by ${h}(x)$ the herd immunity threshold measured as a \textit{fraction} of the number of susceptible nodes. I.e., $h(x) = \nicefrac{H(x)}{(N - x)}$. 
Working with fractions will be useful in the comparison of evolutions with different number of anti-virus deployments (e.g., $x$, $x+1$ and $x+2$). 
Thus, in order to show that Eq. (\ref{maineq}) holds, we will prove that:
	\begin{equation*}
	(N - x) \cdot {h}(x) - 	(N - (x+1))  \cdot {h}(x+1) \ge
\end{equation*}
\begin{equation}\label{eqmarginalvac}
	\ge (N - (x+1)) \cdot {h}(x+1)  - (N - (x+2)) \cdot {h}(x+2).
\end{equation}

For convenience of tracking the overall process, we will add $x$ "virtual" vaccination steps to the list of steps, in order to couple the \textit{vaccination evolution} to the \textit{natural evolution}. Specifically, the indexing of the vaccination process is as follows: 
Steps $1, \dots, x$ correspond to the process of deploying the $x$ vaccines (the nodes moved directly from state \textbf{S} to state \textbf{V}, without being infective.);
Steps $x+1,\dots$ correspond to the infection process afterwards (one step per infected node).

We denote by $\hat{R}(x; n)$ the value of the effective reproduction number at step $n$ of the \textit{vaccination evolution} (i.e., given that $x$ vaccines were deployed in the network). 	

Since the vaccines are given in the region in a random fashion, the \textit{spreading distributions} $\mathcal{S}$ and $\mathcal{I}$ of the susceptible population at steps $1,\dots,x$, remains unchanged while only the \textit{size} of the susceptible population changes. Thus, the reproduction number will reduce linearly for any $0 \le n \le x$:
\begin{equation}\label{eq:rhatandr}
	\hat{R}(x; n) = \frac{N - n}{N} \cdot R_0.
\end{equation}
Figure \ref{Sec5Fig1} depicts the effective reproduction number under a natural evolution of the attack (i.e., no vaccinations), $R()$, vis-a-vis its evolution, $\hat{R}()$, under the vaccination evolution.

Note that $\hat{R}()$ (green) behaves to the right of $x$ exactly as $R()$ (blue) behaves to the right of $0$. Thus we can derive the value of $\hat{R}()$, for $n \ge x$,  by mapping it to the corresponding values of $R()$. We will use this property to analyse $\hat{R}$ in the sequel. This can be done by observing that the number of the initial susceptible nodes addressed by $R()$ is $N$ while that addressed by $\hat{R}()$ is $N-x$, and that the initial reproduction value of $R()$ is $R_0$ while that of $\hat{R}()$ (at $x$) is $(\frac{N - x}{N}) \cdot R_0$, following Eq. (\ref{eq:rhatandr}). Thus it holds that for any $0 \le p \le 1$: 
	
	\begin{equation}
		\hat{R}(x; x + (N - x) \cdot p) =  \left(\frac{N - x}{N}\right) \cdot R(N \cdot p).
	\end{equation}

\begin{remark}
	By its definition, the {exact} value of the Herd Immunity Threshold is {non-discrete} though the infection process is discrete. Hence we address $R()$ and $\hat{R}()$ as continuous functions. Under the assumptions that $N$ is large, such treatment may inflict infinitesimal inaccuracy, but makes the analysis and its description fluent. 
\end{remark}

	\begin{figure}[h]
		\centering
		\includegraphics[width=0.9\linewidth]{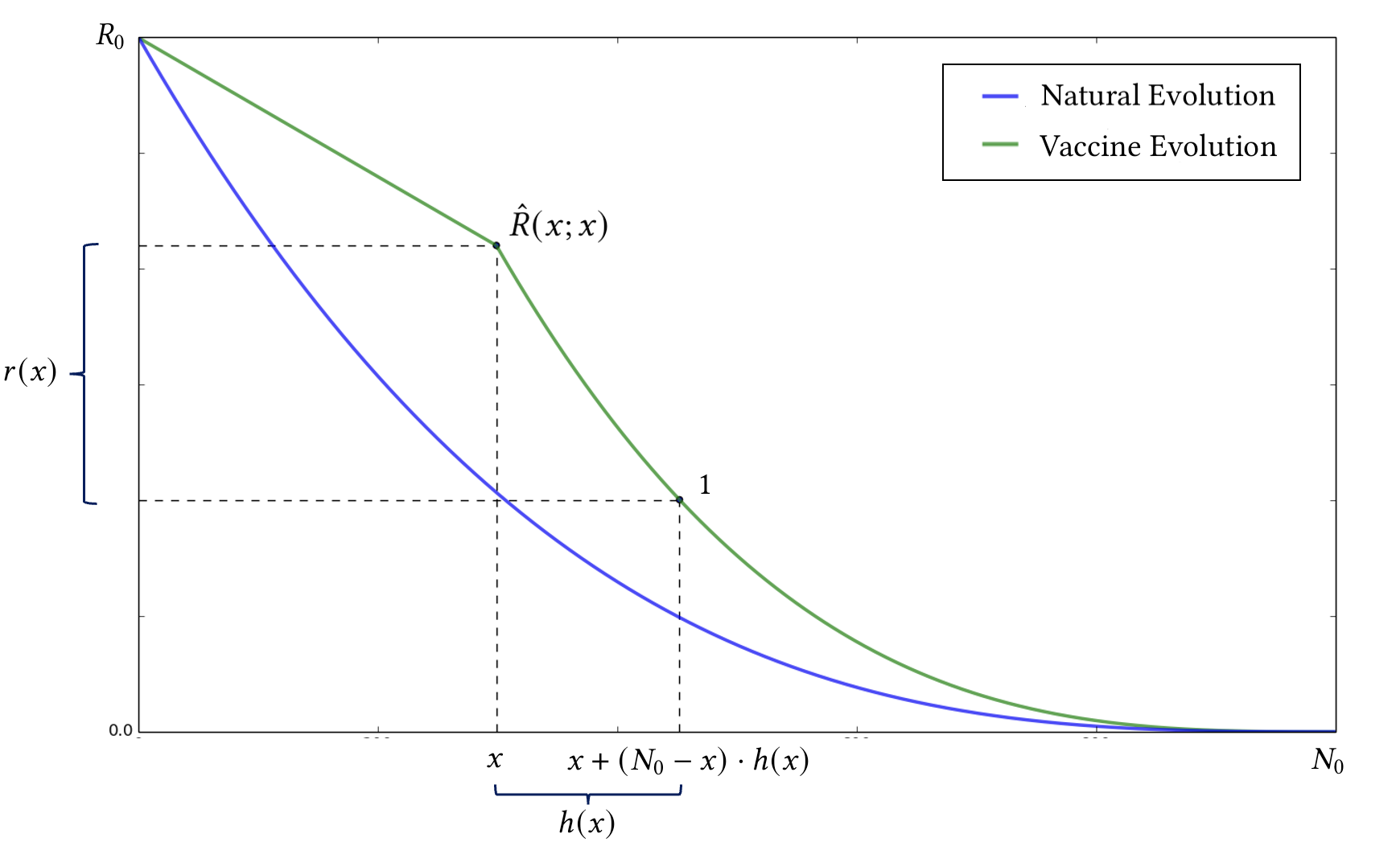}
		\caption{The value of $R(n)$ under a natural evolution vis-a-vis $\hat{R}(x;n)$ under the vaccination evolution, as a function of $n$.}
		\label{Sec5Fig1}
	\end{figure}

\noindent Denote by ${r}(x)$ the relative {drop} observed on $\hat{R}$ when it drops to $1$ (i.e. to reach the herd immunity). See Figure \ref{Sec5Fig1}. That is, 
	\begin{equation}\label{eq:clm1rdef}
		r(x) := \frac{\hat{R}(x;x) - 1}{\hat{R}(x;x)}. 
	\end{equation}
	We next establish two {supporting claims}: 
	\begin{claim}\label{clmhetdratior}
		For any $x$, 
		\begin{equation}\label{eq:clmhetd}
			\frac{h(x+1)-h(x+2)}{h(x)-h(x+1)}
			\le
			\frac{r(x+1)-r(x+2)}{r(x)-r(x+1)}.
		\end{equation}
	\end{claim}
	\begin{claim}\label{lem:randnratio}
	For any $x$,
	\begin{equation}\label{eq:lemrandnrat}
		\frac{r(x+1)-r(x+2)}{r(x)-r(x+1)}=\frac{N-x}{N-x-2}.
	\end{equation}
\end{claim}	
	\begin{proof}[Proof of Claim \ref{clmhetdratior}]
		Following Theorem 1 (the convexity of $R()$), and the definition of a convex function \cite{wiki:Convex_function}\footnote{$f$ is a convex function if and only if $\frac{f(x_2) - f(x_1)}{x_2 - x_1}$ is monotonically non-decreasing in $x_1$ for every fixed $x_2$ (or vice versa).},
		\begin{equation}\label{eq:clm1eq}
			\frac{R(N{h}(x)) - R(N{h}(x+1))}{N{h}(x) - N{h}(x+1)} 
		\le	 \frac{R(N{h}(x+2)) - R(N{h}(x+2))}{N{h}(x+2) - N{h}(x+1))} .
		\end{equation}
		By the mapping between $\hat{R}()$ and $R()$, we know that once a fraction $h(x)$ of the population (of size $N$) is infected the value of ${R}()$ shrinks by a factor of $(1-r(x))$ (see Figure \ref{Fig:randd}). I.e., $R(N \cdot h(x) ) = R_0 \cdot (1-r(x))$. 
Plugging it in the left hand side of Eq. (\ref{eq:clm1eq}) we get: 
		\begin{equation}\label{eq:clm1eqq}
			\frac{R(N\cdot h(x)) - R(N\cdot h(x+1))}{N\cdot h(x) - N\cdot h(x+1)} = 
		\end{equation}
		\begin{equation*}
			= \frac{R_0  (1 - r(x)) - R_0  (1 - r(x+1)) }{N\cdot h(x) - N\cdot h(x+1)} = 
			\frac{R_0}{N} \cdot \frac{ r(x+1) - r(x) }{h(x) - h(x+1)}.
		\end{equation*}
		Deriving Eq. (\ref{eq:clm1eqq}) for $x+1$ as well, combined with Eq. (\ref{eq:clm1eq}):
		\begin{equation}
			\frac{R_0}{N} \cdot \frac{ r(x+1) - r(x) }{h(x) - h(x+1)} \le
			\frac{R_0}{N} \cdot \frac{ r(x+2) - r(x+1) }{h(x+1) - h(x+2)}.
		\end{equation}
		Therefore,
Eq. (\ref{eq:clmhetd}) holds and Claim \ref{clmhetdratior} follows.
	\end{proof}
	Note that for homogeneous populations, the inequality in Eq. (\ref{eq:clmhetd}) is in fact an equality.
	
	\begin{figure}[h]
		\centering
		\includegraphics[width=0.97\linewidth]{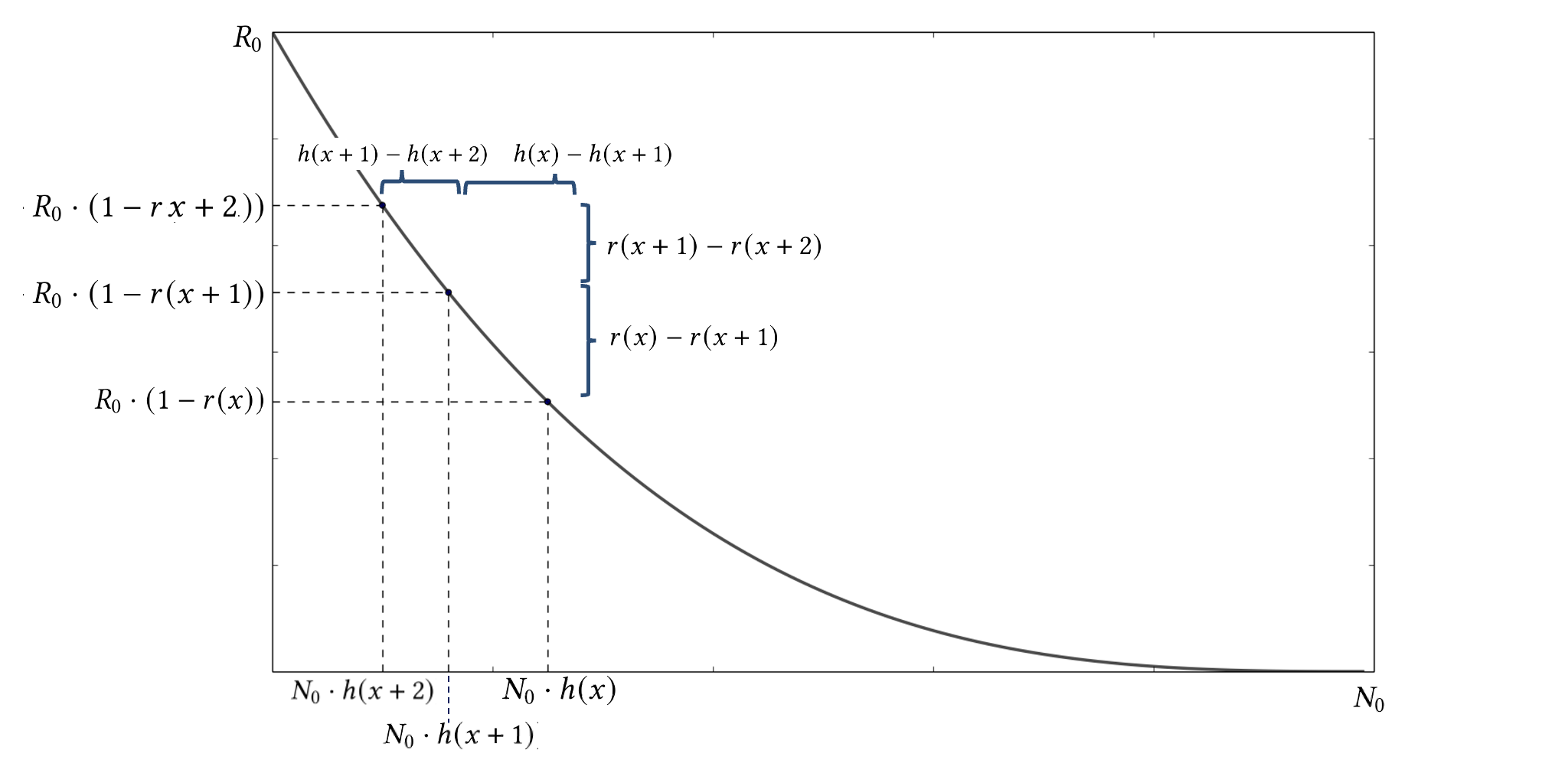}
		\caption{Demonstrating the matching of $\hat{R}$ to $R$.}
		\label{Fig:randd}
	\end{figure}
	\begin{proof}[Proof of Claim \ref{lem:randnratio}]
	It holds that:
	$$
	r(x)-r(x+1)
	=1-\frac{1}{\hat{R}(x;x)}-\left(1-\frac{1}{\hat{R}(x+1;x+1)}\right)=
	$$
	\[
	\frac{1}{\hat{R}(x+1;x+1)}-\frac{1}{\hat{R}(x;x)}=
	\frac{1}{R_{0}\cdot \frac{N-(x+1)}{N}}-\frac{1}{R_{0} \cdot \frac{N-x}{N}}=
	\]
	\[
	\frac{N(N-x)-N(N-x-1)}{R_{0}\cdot(N-x-1)(N-x)}=
	\frac{N}{R_{0}\cdot(N-x-1) \cdot (N-x)}.
	\]
	Similarly,
	$$
	r(x+1)-r(x+2)=\frac{N}{R_{0}\cdot(N-x-2) \cdot (N-x-1)}.
	$$
	Therefore:
    $$		\frac{r(x+1)-r(x+2)}{r(x)-r(x+1)}=
		\frac{\frac{N}{R_{0}\cdot(N-x-2) \cdot (N-x-1)}}{\frac{N}{R_{0}\cdot(N-x-1) \cdot (N-x)}}=
		\frac{N-x}{N-x-2},$$
		and we conclude the proof. 
\end{proof}
	
	\noindent
	Using Claim \ref{clmhetdratior} and Claim \ref{lem:randnratio},
	\begin{equation}\label{eq:nratiogr}
		\frac{h(x+1)-h(x+2)}{h(x)-h(x+1)}\le\frac{N -x}{N - x-2}.
	\end{equation}
	Reorganizing Eq. (\ref{eqmarginalvac}) and then using algebraic manipulation combined with Eq. (\ref{eq:nratiogr}) yields Eq. (3) and concludes the proof of Theorem \ref{thm:hitconvex}.
\end{proof}
\begin{remark}
Note that for homogeneous networks, the inequalities in Eq. (\ref{eqmarginalvac}) is in fact an equality (i.e., linearity holds).
\end{remark}

\subsection{{The Multi-Region Allocation Problem}}\label{secvaccineconvex}
{The system consists of $k$ regions of sizes $N_1, \dots, N_k$. Given their spreading distributions $\mathcal{S}_1,\mathcal{I}_1,\dots,\mathcal{S}_k,\mathcal{I}_k$ and  $v$ vaccines, the operator of the system is required to derive an allocation ${V}=(v_1,\dots,v_k)$ s.t. $\Vert {V} \Vert = v$ while trying to minimize total number of nodes getting infected prior to reaching herd immunity,
$ \sum_{i=1}^{k} H_i(v_i)$.}

{One can derive the optimal allocation using a \textit{na\"ive approach}, by going over all the possible allocations for the system. However, the number of possible allocations is the number of solutions to:
$v_1 + \dots + v_k = v$ such that for any $i$,
$0 \le v_i \le N_i$. 
This exhaustive search derives the optimal allocation at time complexity which is exponential in the number of regions - at least $O(v^k)$ (just to scan all solutions, without accounting to the HIT calculations).
}

{According to Theorem \ref{thm:hitconvex}, our objective function ($ \sum_{i=1}^{k} H_i(v_i)$) is a sum of convex functions.
We propose Algorithm \ref{alg:ALGDPgeneral}, a \textit{greedy} algorithm which derives the optimal allocation. Algorithm \ref{alg:ALGDPgeneral} is similar to known greedy convex optimizing algorithms from the literature (see e.g., \cite{ibaraki1988resource}).}

{
The algorithm iteratively calculates the (absolute value of the) marginal revenue resulting from adding a single vaccine, denoted by $\Delta_i(j)$ (i.e., allocating $j$ vaccines, instead of $j-1$, in region $i$)\footnote{In order to perform the $\Delta_i()$'s calculations, we will use the approximation presented in \cite{oz2021heterogeneity} which allows given the network size, $\mathcal{S}$ and $\mathcal{I}$, to derive $R$ as a function of $n$, and to obtain the expected HIT.}.
	Algorithm 1 maintains for each region $i$ the number of vaccines allocated so far, $v_i$, 
	derives the $\Delta()$ values of the candidates for selection, $\Delta_i(v_i+1)$, 
	adds a vaccine in the region that maximizes the revenue $\Delta_i()$, 
and appends to the candidates set the next vaccine in the selected region.
	The algorithm terminates when the number of allocated vaccines, $\sum v_i$, equals the available number of vaccines, $v$.}

{The optimality can be proven directly from the convexity property (Theorem \ref{thm:hitconvex}), which yields that the marginal revenue from each additional vaccine allocated to a region is decreasing. Thus,
for any $i,j$: $\Delta_i(j) \ge \Delta_i(j+1)$. Hence, each element selected by the algorithm can be part of any optimal allocation and thus its selection is non-regrettable.}

	{The algorithm performs at most $k + v$ $\Delta()$'s calculations ($k$ in the initialization, and after each selection of the maximal element, for $v$ selections).
	The computation of the $\Delta()$s involves the derivation of a weighted expected value on the susceptibility distribution, whose retrieval time is linear at the sizes of the populations $O(\max \{N_1,\dots,N_k\})$ (as their support size can be limited by the population size).

	Therefore, we have a total running time of $O((\max \{N_1,\dots,N_k\}) \cdot (k + v))$.}

		\begin{algorithm}\label{alg2}
		\caption{Optimal Multi-Region Licences Allocation}
		\label{alg:ALGDPgeneral}
		\begin{algorithmic}
			\State \textbf{INPUT:} Population sizes $(N_1, \dots, N_k)$; Spreading Distributions $((\mathcal{I}_1, \mathcal{S}_1), \dots, (\mathcal{I}_k, \mathcal{S}_k))$; $v$ vaccines;
			\State 1. $ (v_1, \dots, v_k) \leftarrow (0, \dots, 0)$.
			\State 2. Calculate $\Delta_1(1), \dots, \Delta_k(1)$.
			\ForAll {$j=1$ to $v$}
			\State a. Choose $i$ s.t. $\Delta_i(v_i + 1) = \max_j \Delta_j(v_j + 1)$.
			\State b. $v_i = v_i + 1$.
			\State c. Calculate $\Delta_i(v_i + 1)$,
			\EndFor
			\textbf{RETURN:} $\textbf{V} = (v_1, \dots, v_k)$.
		\end{algorithmic}
	\end{algorithm}

\section{{The effects of vaccinations during the spreading process}}\label{sec:sec4}

{In this section we address the "online" problem, where the vaccination process, or part of it, occurs during the spread of the virus.}
{In Subsection \ref{sec:sub41} we analyze the interleaving progress of the spreading and the vaccination processes. Although they have different nature (the virus spreading is biased towards "catching" heavy spreaders while the vaccination process is conducted uniformly), we manage to establish a monotonicity property of vaccinations: given a batch of anti-virus installations, the number of infections prior to reaching herd immunity is increasing in their deployment timing.
Using the monotonicity property, in Subsection \ref{sec:sub42} we apply the convexity properties from Section \ref{sec:sec3} and derive optimal online deployment policy.}

\subsection{Monotonicity of The HIT in The Vaccination Timing}\label{sec:sub41}

{We consider a network consisting of $N$ nodes, to which a batch of $x$ vaccines are allocated ($x \ge 1$). We compare two scenarios: deploy the $x$ vaccines at step $i_1$ (i.e., after $i_1$ infections have occurred), or deploy them at step $i_2$ (after $i_2$ infections have occurred), where $i_1 < i_2$. We will prove that:

\begin{theorem}\label{thm:vacctime}
Given $x$ vaccines, the number of infected individuals prior to reaching herd immunity is monotonically non-decreasing in the step at which the vaccines are deployed.  
\end{theorem}

I.e., deploying the given allocation as soon as it is available is the optimal policy.

The proof follows similar ideas to those of the proof of Theorem 2 from Section III. The main difference is that now, the compared evolutions have same allocation size but different timing of deployment (while in Section III same timing, prior to the spread, was assumed with different allocation sizes).

The indexing of the vaccination process is as follows: 
Steps $1, \dots, i$ correspond to the steps of the  infection process occurring before the anti-virus are deployed (i.e., one step per infection of a node);
Steps $i+1, \dots, i+x$ correspond to the process of deploying the $x$ anti-viruses; 
Steps $i+x+1,\dots$ correspond to the infection process afterwards.

We denote by $\hat{R}(i,x;n)$ the value of the effective reproduction number at step $n$, assuming that {$x$ vaccines are deployed at steps $i+1, \dots, i+x$}. In order to prove Theorem 3, we will establish the following Claim:

\begin{claim}\label{clm:vacctime}
	Let $i_1 < i_2$, then for any $n \ge x + \max\{i_1, i_2\}$,
	\begin{equation}
		\hat R(i_1,x;n) \le \hat R(i_2,x;n).
	\end{equation}

\end{claim}

\begin{proof}[Proof of Claim \ref{thm:vacctime}]

Note that at step $i+x$ of the \textit{vaccination evolution} the susceptible population (of size $N - i - x$) has the same spreading distribution as the susceptible population (of size $N - i$) at step $i$ of the \textit{natural evolution}, due to the random fashion of the vaccination.

Figure \ref{fig:vacc_notation} depicts the effective reproduction number under a natural evolution of the spreading process (i.e., no vaccines), $R$, vis-a-vis its evolution under the vaccination evolution, $\hat{R}$. 
Hence, the shape of $R()$ (in blue) to the right of $i$ is identical to the shape of $\hat{R}()$ (green) to the right of $i+x$, but with a different scale.

\begin{figure}[h]
	\centering
	\includegraphics[width=0.85\linewidth]{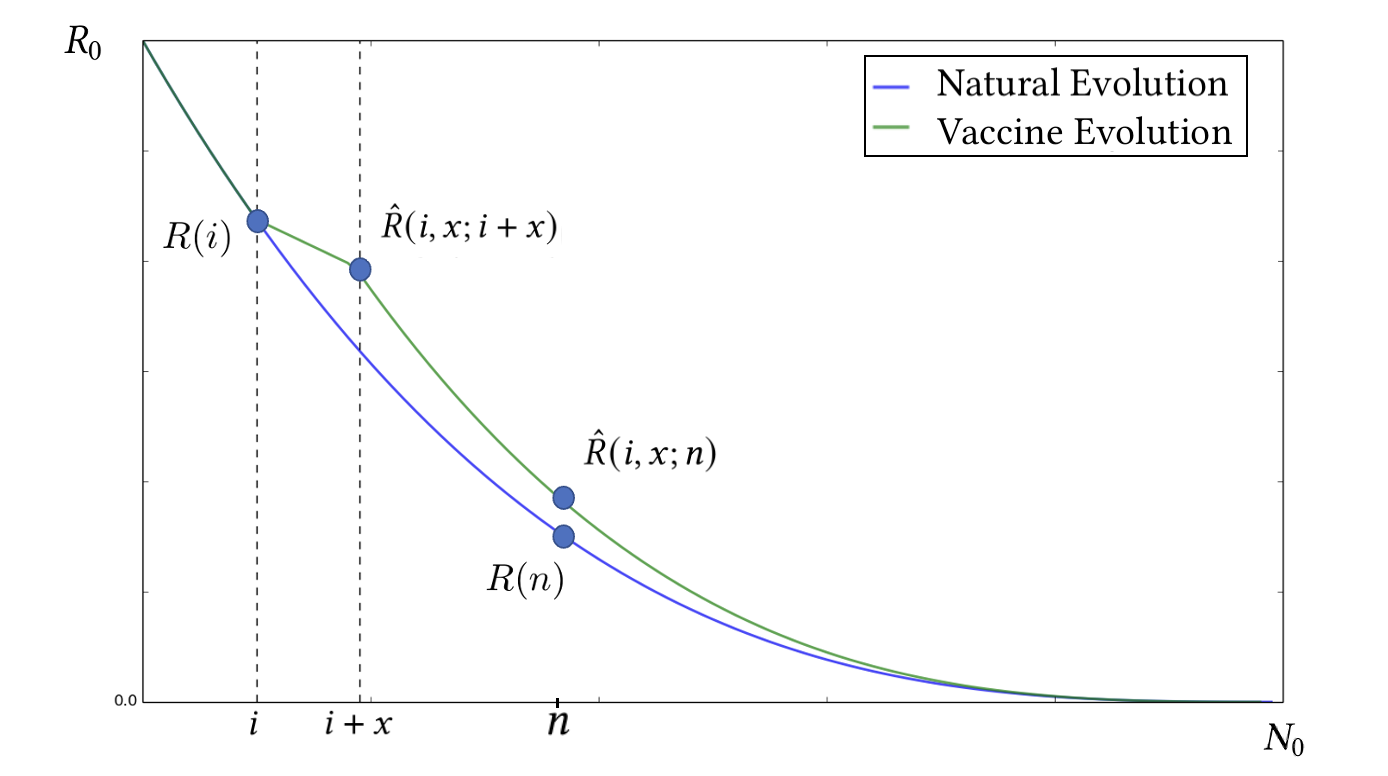}
	\caption{The value of $R$ under a natural evolution vis-a-vis $\hat{R}$ under the vaccination evolution, as a function of $n$.}
	\label{fig:vacc_notation}
\end{figure}

We establish the following relation between $R$ and $\hat{R}$.

\begin{lemma}\label{lemrvactimet}
	Consider $x,i$ and $n$ be such that $n > i + x$. Then:
	\begin{equation}\label{eq:rvactimedef}
		\hat{R}(i,x;n) = R\left( n - \frac{x \cdot (N - n)}{N - i - x} \right)\cdot\left(1-\frac{x}{N-i}\right).
	\end{equation}
\end{lemma}

\begin{proof}[Proof of Lemma \ref{lemrvactimet}]
	Recall that the spreading distributions of the population at step $i+x$ of the {vaccine evolution}, and at step $i$ in the natural evolution are the same, since the vaccines are given in a random fashion. Yet, the sizes of the susceptible populations in these evolutions are $N - (i + x)$ and $N - i$, respectively.
	Since $R(n)$ is determined by the size and the distribution of the susceptible population \cite{oz2021heterogeneity}, it holds that
	\begin{equation}\label{eqrhatdef}
		\hat{R}(i,x;i+x)=\frac{N-(i+x)}{N-i} \hat{R}(i,x;i) =\frac{N-(i+x)}{N-i} R(i).
	\end{equation}
	{The last equality holds since the systems separate only after step $i$ (at which the vaccines are deployed), and hence $\hat{R}(i,x;i) = R(i)$.}
	Following similar reasoning, since the systems are progressing by the same spreading functions with different scales (of $N - i$ and $N - (i + x)$),
	for any $0<p<1$,
	\begin{equation}\label{eq:pinr}
		\frac{\hat{R}(i,x \; ; \; (i+x)+p\cdot(N-(i+x)))}{\hat{R}(i,x;i+x)}=\frac{R(i+p\cdot(N-i))}{R(i)}.
	\end{equation}	

	For any $i + x < n < N$ it holds that:	$
	n = (i+x)+\frac{n-(i+x)}{N-(i+x)}\cdot(N-(i+x))
	$). Thus,
	\begin{equation}\label{myequp}
	\hat{R}(i,x;n)
	=\hat{R}\left(i,x  ;  (i+x)+\frac{n-(i+x)}{N-(i+x)}(N-(i+x))\right).
	\end{equation}
	Note that $0 < \frac{n-(i+x)}{N-(i+x)} < 1$. Hence, by Eq. (\ref{eq:pinr}) and (\ref{myequp}):
	$$
	\hat{R}(i,x;n) = \frac{R\left(i+\frac{n-(i+x)}{N-(i+x)}\cdot(N-i)\right)}{R(i)}\cdot\hat{R}(i,x;(i+x)).
	$$
	Substituting Eq. (\ref{eqrhatdef}) into this expression we get: 
	\begin{equation}\label{eq:forappend}
		\hat{R}(i,x;n) = \frac{R\left(i+\frac{n-(i+x)}{N-(i+x)}\cdot(N-i)\right)}{R(i)}\cdot\frac{N-(i+x)}{N-i}\cdot R(i)=
	\end{equation}

	$$
	=R\left(i+\frac{n-(i+x)}{N-(i+x)}\cdot(N-i)\right)\cdot\frac{N-(i+x)}{N-i}
	=
	$$
	\\

	$$
	=
	R\left(i+\frac{Nn-Ni-xN-in+i^{2}+xi}{(N-i-x)}\right)\cdot \left(1-\frac{x}{N-i}\right)=
	$$

	$$
	=
	R\left(i+\frac{Nn-xN-in}{(N-i-x)}-i\right)\cdot\left(1-\frac{x}{N-i}\right) =
	$$

	$$
	=R\left(\frac{Nn-in-xN+(xn-xn)}{(N-i-x)}\right)\cdot\left(1-\frac{x}{N-i}\right)=
	$$

	$$
	=R\left(n-\frac{x\cdot (N-n)}{N-i-x}\right)\cdot\left(1-\frac{x}{N-i}\right). 
	$$
\end{proof}

\begin{figure}
	\centering
	\includegraphics[width=0.85\linewidth]{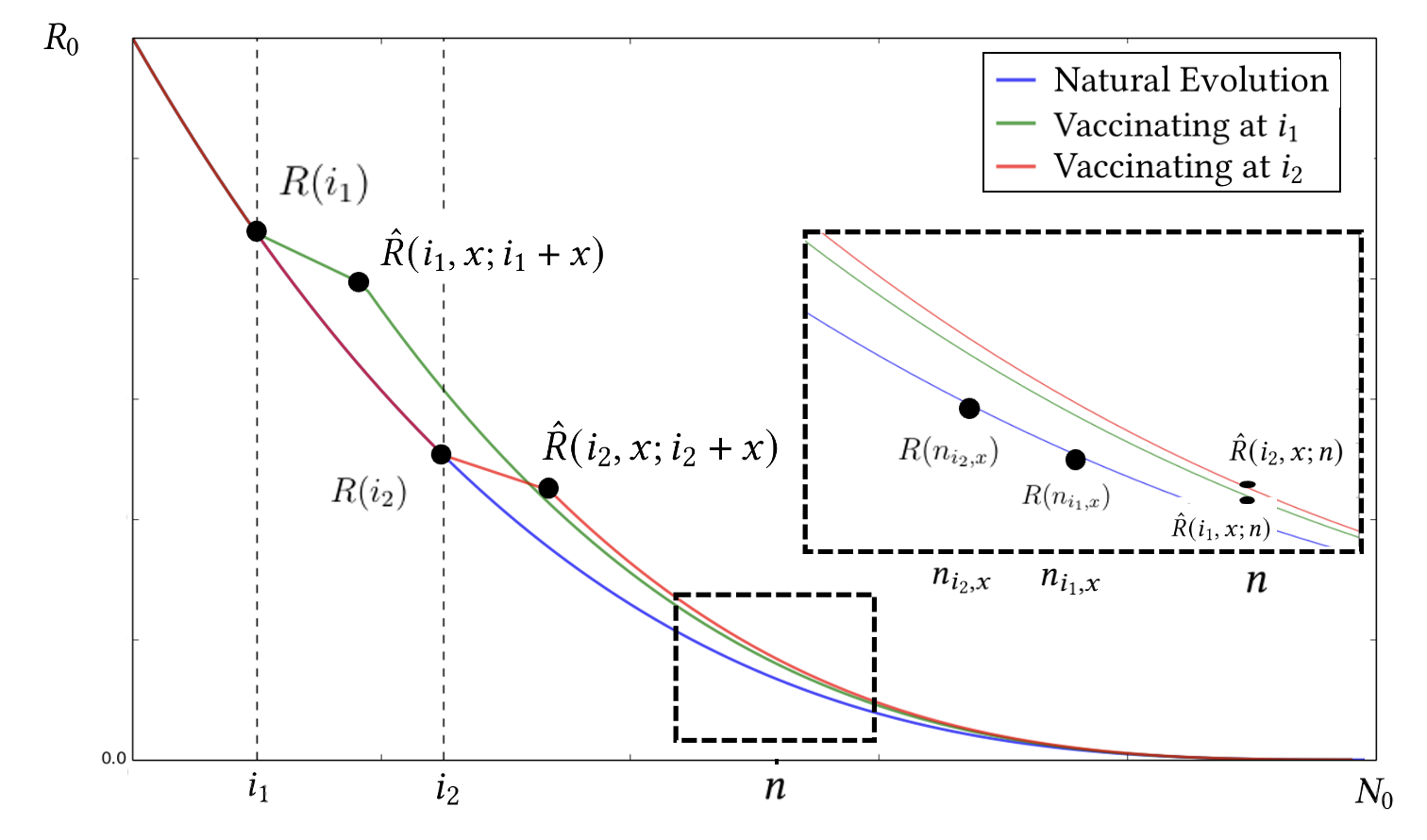}
	\caption{The effective reproduction number as a function of $n$, comparing deploying $x$ vaccines at step $i_1$ and at step $i_2$.}
	\label{figR_vac_illus}
\end{figure}
\noindent
For convenience, we denote  
$$	d_{i,x} := 1 - \frac{x}{N-i}
\; 
\;
\text{and}
\;
\;
n_{i,x} := n - \frac{x \cdot (N - n)}{N - i - x}.$$
By Lemma \ref{lemrvactimet}, for any $n > x+i$:
$ \hat{R}(i,x;n) = R(n_{i,x}) \cdot d_{i,x}.$

\begin{lemma}\label{helplem2}
	Let $x$ and $i_1, i_2$. Then:
	\begin{equation}
		\frac{N - n_{i_1,x}}{N - n_{i_2,x}} = \frac{d_{i_2,x}}{d_{i_1,x}}.
	\end{equation}
\end{lemma}

\begin{proof}[Proof of Lemma \ref{helplem2}]
	It holds that,
	\begin{equation}\label{eq:appeneq1}
		\frac{N-n_{i_1,x}}{N-n_{i_2,x}} 
		=
		\frac{\frac{N-n_{i_1,x}}{N}}{\frac{N-n_{i_2,x}}{N}}
		=
		\frac{\left(1-\frac{n_{i_1,x}}{N}\right)}{\left(1-\frac{n_{i_2,x}}{N}\right)}.
	\end{equation}
	Note that,
	$$
	\left({1-  \frac{n_{i,x}}{N}}\right) \cdot d_{i,x}  
	= \left(1-\frac{\left(n-\frac{xN-xn}{N-i-1}\right)}{N}\right) \cdot \frac{N-i-x}{N-i}=
	$$
	$$
	=\left(\frac{N-i-x}{N-i}-\frac{n(N-i-x)}{N(N-i)}+\frac{(xN-xn)}{N(N-i)}\right)=\left(\frac{N-n}{N}\right).
	$$
	I.e., the product $\left({1-  \frac{n_{i,x}}{N}}\right) \cdot d_{i,x}$ is constant.
That is,
	$$
	d_{i_1,x}\cdot\left(1-\frac{n_{i_1,x}}{N}\right)=d_{i_2,x}\cdot\left(1-\frac{n_{i_2,x}}{N}\right).
	$$
	Hence by Eq. (\ref{eq:appeneq1}), 
	$$ \frac{(N-n_{i_1,x}) \cdot d_{i_1,x}}{(N-n_{i_2,x}) \cdot d_{i_2,x}} = \frac{\left(1-\frac{n_{i_1,x}}{N}\right) \cdot d_{i_1,x}}{\left(1-\frac{n_{i_2,x}}{N}\right) \cdot d_{i_2,x}} = 1 $$
	and therefore 
	$$
	\frac{N-n_{i_1,x}}{N-n_{i_2,x}} = \frac{d_{i_2,x}}{d_{i_1,x}},
	$$ as required.
\end{proof}

Having established Lemma \ref{lemrvactimet} and Lemma \ref{helplem2}, we now conclude the proof of Claim \ref{thm:vacctime}:
	Let $x$, and let $0 < i_1 < i_2 < N$. 
	By Lemma \ref{lemrvactimet} we have that $\hat{R}(i_1,x;n)= R ( n_{i_1,x}) \cdot d_{i_1,x}$, and  $\hat{R}(i_2,x;n)= R ( n_{i_2,x}) \cdot d_{i_2,x}$.

	Note that $n_{i_2,x} < n_{i_1,x}$ (see Figure \ref{figR_vac_illus}), and thus by the convexity of $R$ (Theorem 1):
	\begin{equation}\label{eq:rratioinN}
		\frac{R(n_{i_1,x})}{R(n_{i_2,x})}
		\le \frac{N-n_{i_1,x}}{N-n_{i_2,x}}.
	\end{equation}

	By Eq. (\ref{eq:rratioinN}) and Lemma \ref{helplem2}.
	$	\frac{R(n_{i_1,x})}{R(n_{i_2,x})} \le \frac{d_{i_2,x}}{d_{i_1,x}}$,
Thus,
	\begin{equation*}
		\hat{R}(i_1,x;n) = R(n_{i_1,x})\cdot d_{i_1,x} \le R(n_{i_2,x})\cdot d_{i_2,x} = \hat{R}(i_2,x;n).
	\end{equation*}
\end{proof}

	Since the herd immunity is reached once $R$ drops below 1, Theorem \ref{thm:vacctime} follows immediately from Claim \ref{clm:vacctime}.

\subsection{Online Policy: Optimal Allocation and Timing}\label{sec:sub42}
According to Theorem \ref{thm:vacctime}, the number of infected nodes prior to reaching herd immunity in the network is increasing in the step at which the vaccines are deployed. The following corollary follows directly from the monotonicity property,
\begin{corollary}
{An optimal deployment policy which minimizes the HIT deploys each batch of vaccines at the earliest step at which they are available.}
\end{corollary}
Due to the objective function being additive on the regions ($\sum_{i=1}^{k}H_i(v_i)$), by applying the corollary in each of the regions an optimal solution for the combined problem of allocation and timing will be to deploy each batch of vaccines once it is available, and to allocate it to the regions once it is available according to Algorithm \ref{alg:ALGDPgeneral}.

\section{{Numerical Evaluations and Experiments}}\label{sec:sec5}
We use numerical evaluations to study the sensitivity of the performance of the attacked systems to the heterogeneity level of the different regions of the system, and to the chosen deployment strategy. 
{We use \textit{gamma} and \textit{power-law} degree distributions\footnote{{Many large real-world networks, including the WWW \cite{albert1999diameter}, the Internet \cite{faloutsos1999power}, and other computer networks \cite{tilovska2011performance, nguyen2011analysis} have shown to follow a power-law tail degree distribution; i.e., for large degree $k$, the fraction of nodes having this degree is $P(k) \sim k^{-\lambda}$. The parameter $\lambda$ is called the power-law exponent and its value is typically in the range $2 < \lambda < 3$.}}, with varying shape or exponent parameters, in order to inspect varying heterogeneity levels.}
	
We begin with focusing on a single region of the system.
Figure \ref{fig:discus}(a)depicts the marginal revenue (measured in reduction of the  infections number) gained from allocating an additional vaccine to the region\footnote{In fact, note that Figure \ref{fig:discus}(a) depicts the $\Delta()$ values from Algorithm \ref{alg:ALGDPgeneral}}. The marginal revenue decreases as the allocation size increases, as expected (Section \ref{sec:sec3}); Figure \ref{fig:discus}(b) depicts the effect of the deployment timing and asserts that any delay results with an increase in the number of infection cases, as expected (Section \ref{sec:sec4}). Note that as the heterogeneity level of the region increases, the timing effect increases.

\begin{figure}[h]
	\centering
	\includegraphics[width=0.85\linewidth]{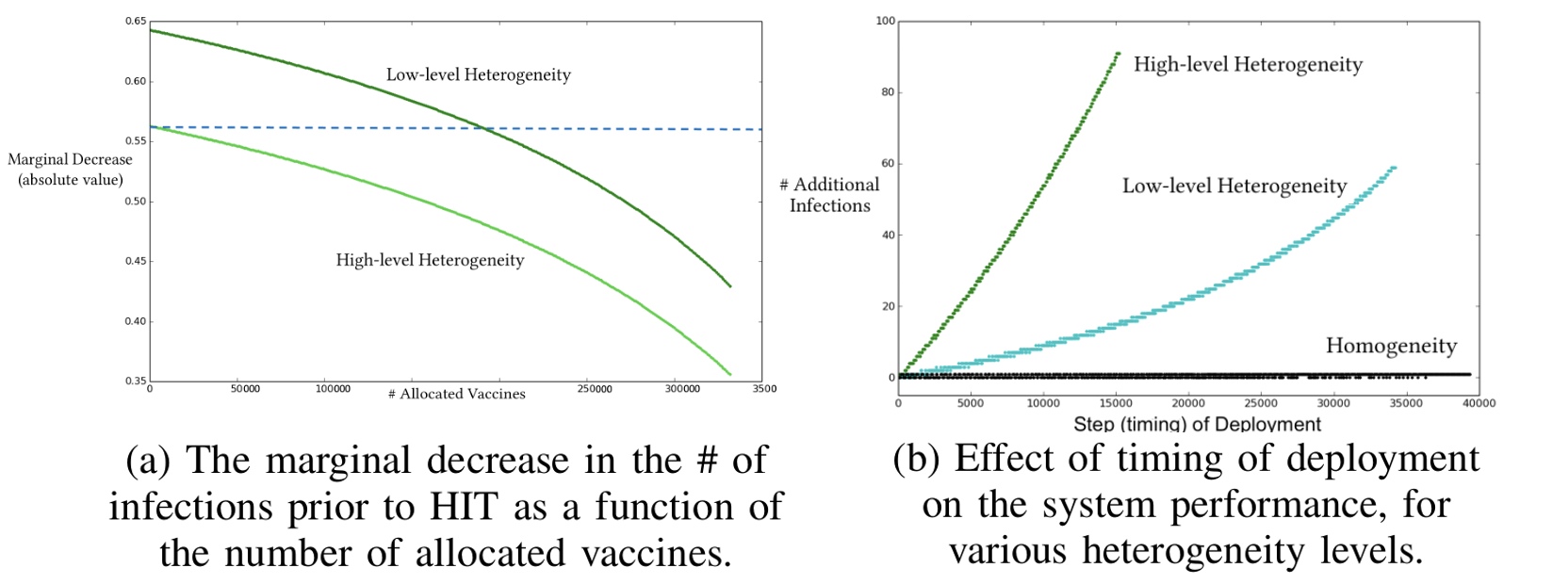}
	\caption{Inspecting deployment in a single region of the system.}
	\label{fig:discus}
\end{figure}
\noindent
We next move to inspect the effect of the allocation strategy on the resultant HIT and overall number of infections in a bi-regional system with equal region sizes (number of nodes).

We repeat the experiment for two cases: (1) The distribution of both regions is sampled according to the Barabási–Albert model (i.e., both regions follows power law of the form $P(k)\sim k^{{-3}}\,$) \cite{albert2002statistical}; (2) Region 1 follows a truncated power law distribution with exponent $\lambda = 1.5$, while Region 2 has $\lambda = 4.5$ (the spreading values are normalized such that both regions posses the same {mean} spreading degree and $R_0$ while having different variance and distributions).

\begin{figure}[h]
	\centering
	\includegraphics[width=0.85\linewidth]{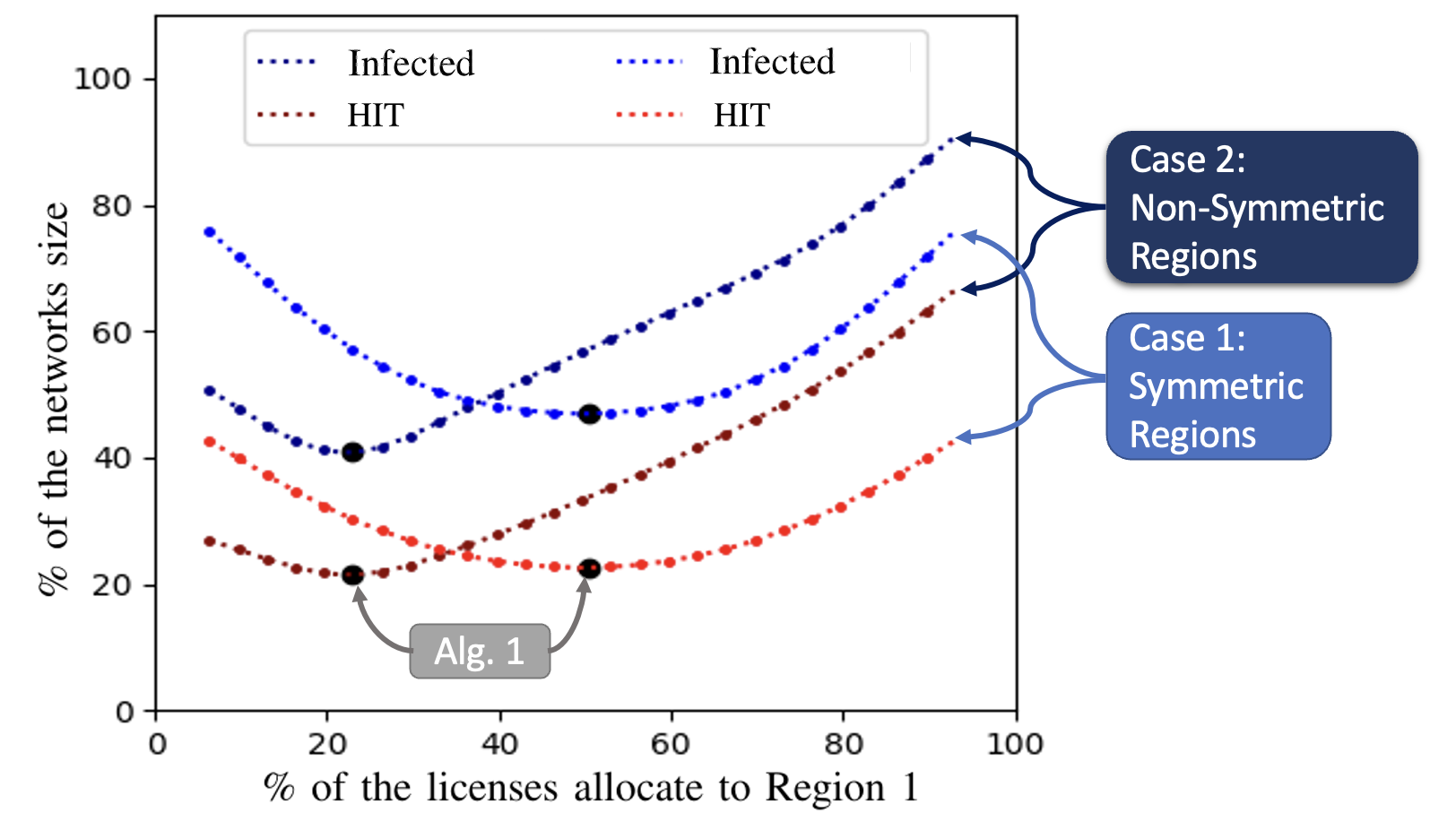}
	\caption{The values (in percentage) of the HIT and total number of infections in a bi-regional system as a function of the allocation (\% of the available licenses allocated to Region 1).}
	\label{fig:discuss_sens}
\end{figure}

Figure \ref{fig:discuss_sens} depicts the HIT and the overall number of infections as a function of the allocation. The allocation is determined by the percentage of available licenses allocated to Region 1 (while the rest are assigned to Region 2). 
In the symmetric case, in which both regions share the same spreading distribution, the optimal allocation is symmetric (the minimum is reached at 50\% on the x-axis), in accordance with the convexity property.
In the non-symmetric case, the minimization of the HIT is reached by allocating a greater amount of licenses to the less heterogeneous region.

The allocations derived by Algorithm \ref{alg:ALGDPgeneral} are marked in black on the figure and those correspond to the minimum value of the HIT in the each case's plot.
Note that those solution minimize not only the HIT, but also the overall number of infection cases (the blue lines).
As can be observed, the system performance is highly sensitive to the heterogeneity level of the population;
the difference in the performance of the optimal and an arbitrary strategies might reach tens of percents.

\section{{Summary}}\label{sec:sec6}
We studied the effect of anti-virus deployment strategies on the progression and damage cause by spreading computer viruses in heterogeneous networks.
In particular, we analyze the effect of the allocation quantity and the deployment timing on the herd immunity acquired by the network.
We proposed a very efficient (greedy) optimal allocation algorithm which allocates anti-virus licensees in multi-region systems, and proposed vaccination deployment strategies both in the offline (prior to the attack) and online (during the attack) scenarios.

The effects of allowing some weak connection between nodes in different regions, in the light of heterogeneity, is being studied in an on going research.

\section{Acknowledgment}
This research was supported in part by the Israel Science Foundation (grant No. 2482/21) and by the Blavatnik Family Foundation.

\bibliographystyle{abbrv}
\bibliography{avallocatiob_noms2023}

\end{document}